\documentclass[11pt]{article}
\usepackage[margin=1in]{geometry}
\usepackage{amsmath,amssymb,amsthm}
\usepackage{bbm}
\usepackage{stmaryrd}
\usepackage{graphicx}
\usepackage{paralist}
\usepackage{latexsym}
\usepackage{color}
\usepackage{varioref}
\usepackage{enumerate}
\usepackage{pdfsync}
\usepackage{enumitem}
\usepackage{hyperref}

\newtheorem{mythm}{Theorem} \numberwithin{mythm}{section}

\newtheorem{mylemma}[mythm]{Lemma}
\newtheorem{mydef}[mythm]{Definition}
\newtheorem{myrek}[mythm]{Remark}
\newtheorem{myassump}[mythm]{Assumption}
\newtheorem{mycor}[mythm]{Corollary}

\DeclareMathAlphabet\scr{U}{scr}{m}{n}
\SetMathAlphabet\scr{bold}{U}{scr}{b}{n}
  \DeclareFontFamily{U}{scr}{\skewchar\font'177}%
  \DeclareFontShape{U}{scr}{m}{n}{<-6>rsfs5<6-8>rsfs7<8->rsfs10}{}%
  \DeclareFontShape{U}{scr}{b}{n}{<-6>rsfs5<6-8>rsfs7<8->rsfs10}{}%

\numberwithin{equation}{section}

\newcommand{\auf}{[\![}
\newcommand{\zu}{]\!]}

\renewenvironment{thebibliography}[1]{%
\begin{oldthebibliography}{#1}%
\setlength{\baselineskip}{.8em}
\linespread{.8}
\small
\setlength{\parskip}{0ex}%
\setlength{\itemsep}{.1em}%
}%
{%
\end{oldthebibliography}%
}

\begin{document}

\title{\vspace{-1.5cm} High-Resilience Limits of Block-Shaped Order Books\footnote{We thank Yan Dolinsky, Martin Herdegen, Sebastian Hermann, and Ren Liu for fruitful discussions.}}
\author{
Jan Kallsen
\thanks{Christian-Albrechts-Universit\"at zu Kiel, Mathematisches Seminar, Westring 383, D-24098 Kiel, Germany, email \texttt{kallsen@math.uni-kiel.de}.}
\and
Johannes Muhle-Karbe\thanks{ETH Z\"urich, Departement f\"ur Mathematik, R\"amistrasse 101, CH-8092, Z\"urich, Switzerland, and Swiss Finance Institute, email \texttt{johannes.muhle-karbe@math.ethz.ch}. Partially supported by the ETH Foundation.}
}
\date{September 25, 2014}
\maketitle

\maketitle

\begin{abstract}
We show that wealth processes in the block-shaped order book model of Obizhaeva/Wang~\cite{obizhaeva.wang.13} converge to their counterparts in the reduced-form model proposed by Almgren/Chriss~\cite{almgren.chriss.99,almgren.chriss.01}, as the resilience of the order book tends to infinity. As an application of this limit theorem, we explain how to reduce portfolio choice in highly-resilient Obizhaeva/Wang models to the corresponding problem in an Almgren/Chriss setup with small quadratic trading costs.
\end{abstract}

\section{Introduction}

Market prices are adversely affected by large orders executed quickly. This ``price impact'' constitutes the principal trading cost for large institutional investors and hedge funds. Accordingly, there is a large and growing literature studying how to mitigate these costs by smart scheduling of the order flow (see, e.g., \cite{gokay.al.12,gatheral.schied.13} for recent overviews).

Two of the most widely used models were proposed by Almgren and Chriss as well as Obizhaeva and Wang, respectively. Almgren and Chriss \cite{almgren.chriss.99,almgren.chriss.01,almgren.03} put forward a \emph{reduced-form} model that directly specifies functions describing the temporary and permanent impacts of a given order.\footnote{Similar models were proposed and studied concurrently by Bertsimas and Lo~\cite{bertsimas.lo.98}, Madhavan~\cite{madhavan.00}, as well as Huberman and Stanzl~\cite{huberman.stanzl.05}. To keep in line with most of the literature, we nevertheless stick to the nomenclature ``Almgren/Chriss-model''.} Concrete specifications of this model are very tractable, for optimal execution \cite{almgren.chriss.01,almgren.03,schied.schoeneborn.09} and -- if price impact is linear -- also for more involved problems like portfolio choice \cite{garleanu.pedersen.13a,garleanu.pedersen.13b, guasoni.weber.12,moreau.al.14} and hedging \cite{almgren.li.11,moreau.al.14}. In particular, fully explicit formulas obtain in the limit for small linear temporary impacts even in very general settings \cite{moreau.al.14}. 

The model of Obizhaeva/Wang \cite{obizhaeva.wang.13} is \emph{structural}, in that it does not directly posit some price impact functions, but instead starts from a description of the underlying limit order book. In today's electronic markets, the latter collects all outstanding limit buy and sell orders and thereby provides a snapshot of the current liquidity on offer. Obizhaeva and Wang assume that the order book is ``block-shaped'', i.e., liquidity is distributed uniformly beyond the best bid and ask prices.\footnote{Compare, e.g., \cite{alfonsi.al.10,predoiu.al.11} for extensions to more general order book shapes.} As each trade eats into the order book, this leads to another model with linear price impact. Yet, after each trade, the book now recovers at some finite resilience rate, leading to a ``transient'' price impact that is neither purely temporary nor fully permanent. While risk-neutral optimal execution problems are still feasible in this setting \cite{obizhaeva.wang.13,alfonsi.al.10,predoiu.al.11}, more complex optimizations have proved to be elusive so far. 

In the present study, we show that the models of Almgren/Chriss and Obizhaeva/Wang are intimately related. Indeed, as the resilience of the order book in the Obizhaeva/Wang model grows, the corresponding wealth processes converge to their counterparts in the Almgren/Chriss setting with small linear temporary price impact. The latter is determined explicitly in terms of observable properties of the order book, namely i) its height, ii) its resilience, and iii) how symmetrically it recovers on the bid and the ask sides.

In view of this limit theorem, optimization problems in highly resilient Obizhaeva/Wang models can be related to their counterparts in Almgren/Chriss models with small linear price impacts. We illustrate this for portfolio choice, using the results of Moreau et al.~\cite{moreau.al.14} in the framework of Almgren/Chriss to obtain asymptotically optimal trading strategies in Obizhaeva/Wang settings.

The remainder of the paper is organized as follows. Section~\ref{sec:trading} introduces a general version of the block-shaped limit order book model of Obizhaeva and Wang, in particular the corresponding wealth dynamics. In Section~\ref{sec:AC}, we turn to reduced-form models of Almgren/Chriss-type. Section~\ref{sec:Asymp} contains our main results showing that Almgren/Chriss-models arise as the high-resilience limits of block-shaped order books. Finally, the implications of this limit theorem for portfolio choice are discussed in Section~\ref{sec:PC}.

\section{Trading with Block-Shaped Order Books}\label{sec:trading}

We fix a filtered probability space $(\Omega,\scr{F}, (\scr{F}_t)_{t \in \mathbb{R}_+},P)$ and consider a general version of the order book model proposed by Obizhaeva and Wang \cite{obizhaeva.wang.13}.\footnote{A similar model with stochastic market depth and resilience has been proposed in \cite{roch.soner.13}.} To wit, there is a safe asset, with price process normalized to $S^0 \equiv 1$, and a risky asset. The \emph{fundamental price process} of the latter is modeled as a continuous semimartingale $S$, traded with exogenous proportional transaction costs $\varepsilon^\downarrow$, $\varepsilon^\uparrow$ for selling and buying, respectively. This means that in the absence of a large trader, infinitesimal amounts can be sold for $S_t-\varepsilon^\downarrow_t$ and purchased for $S_t+\varepsilon^\uparrow_t$ at time $t$, respectively, i.e., $\varepsilon^\downarrow$, $\varepsilon^\uparrow$ are the \emph{unaffected bid-ask spreads}.

In addition to this exogenous friction, large trades also move market prices. Following Obizhaeva and Wang \cite{obizhaeva.wang.13}, this is described by a block-shaped order book with finite volume. To make this precise, consider a \emph{trading strategy} described by an adapted, right-continuous finite variation process $\varphi=\varphi^\uparrow-\varphi^\downarrow$.\footnote{Here, $\varphi^\uparrow-\varphi^\downarrow$ is the decomposition of $\varphi$ into minimal nondecreasing right-continuous processes $\varphi^\uparrow,\varphi^\downarrow$.} The block-shaped order book has height $h_t^\uparrow$ on the buy side, and height $h^\downarrow_t$ on the sell side. Hence, a buy order $d\varphi^\uparrow_t$ increases the best ask price by $\frac{1}{h^\uparrow_t}d\varphi^\uparrow_t$, whereas the current best bid price remains unchanged. After the completion of each trade, the bid-ask prices do no necessarily mean revert back to the fundamental price, but to some \emph{reference price} that also takes into account the permanent impact of past trades. As in \cite{obizhaeva.wang.13}, this permanent impact on the reference price is linear,\footnote{Non-linear permanent price impact typically leads to ``quasi-arbitrage'' due to price manipulation, see Huberman and Stanzl~\cite{huberman.stanzl.04}.} given by a fraction $\alpha^\uparrow_t \in [0,1/2]$ of the move of the best ask price.\footnote{Here, $\alpha^\uparrow \equiv 1/2$ means that the reference quote is the mid price. For $\alpha^\uparrow \equiv 0$, one instead uses the fundamental price $\tilde{S}$ because there is no permanent price impact.} Therefore, a buy order $d\varphi^\uparrow_t$ increases the ask spread $\varepsilon^\uparrow_t$ relative to the reference quote by $\frac{1-\alpha^\uparrow_t}{h^\uparrow_t}d\varphi^\uparrow_t$ and the bid spread $\varepsilon^\downarrow_t$ by $\frac{\alpha_t^\uparrow}{h^\uparrow_t}d\varphi^\uparrow_t$.\footnote{Hence, $\alpha^\uparrow \leq 1/2$ ensures that buy trades do not affect the spread on the bid side more than on the ask side.} After each purchase, the bid and ask spreads mean revert back to the exogenous baseline levels $\varepsilon^\downarrow_t$, $\varepsilon^\uparrow_t$ with resilience rates $\kappa^\downarrow_t$ and $\kappa^\uparrow_t$, respectively. The effects of sell orders are analogous. 

\medskip 

In summary, the reference price follows
\begin{equation}\label{eq:level}
dS^\varphi_t:= dS_t +\frac{\alpha_t^\uparrow}{h_t^\uparrow} d\varphi^\uparrow_t-\frac{\alpha_t^\downarrow}{h^\downarrow_t}d\varphi^\downarrow_t, 
\end{equation}
and the corresponding bid-ask spreads evolve according to
\begin{align}
d\varepsilon^{\varphi,\downarrow}_t &= \kappa^\downarrow_t(\varepsilon^\downarrow_t- \varepsilon^{\varphi,\downarrow}_t) dt + \frac{1-\alpha^\downarrow_t}{h^\downarrow_t}d\varphi^\downarrow_t+\frac{\alpha^\uparrow_t}{h^\uparrow_t}d\varphi^\uparrow_t+ d\varepsilon^\downarrow_t, \quad \varepsilon^{\varphi,\downarrow}_0=\varepsilon^{\downarrow}_0, \label{eq:veps1}\\
d\varepsilon^{\varphi,\uparrow}_t &= \kappa^\uparrow_t (\varepsilon^\uparrow_t-\varepsilon^{\varphi,\uparrow}_t) dt +\frac{1-\alpha^\uparrow_t}{h^\uparrow_t}d\varphi^\uparrow_t+\frac{\alpha^\downarrow_t}{h^\downarrow_t}d\varphi^\downarrow_t+d\varepsilon^\uparrow_t, \quad \varepsilon^{\varphi,\uparrow}_0=\varepsilon^{\uparrow}_0. \label{eq:veps2}
\end{align}

Throughout, we impose the following regularity conditions which imply, in particular, that the reference quote and the corresponding bid-ask spreads are well defined for any trading strategy $\varphi$:

\begin{myassump}
$\kappa^\uparrow,\kappa^\downarrow,h^\uparrow,h^\downarrow$ are positive, $\alpha^\uparrow,\alpha^\downarrow$ take values in $[0,1/2]$, and all of these processes  are continuous and adapted. $\varepsilon^\downarrow$, $\varepsilon^\uparrow$ are positive, continuous semimartingales.
\end{myassump}

For any given trading strategy $\varphi=\varphi^\uparrow-\varphi^\downarrow$, the corresponding bid-ask spreads $\varepsilon^{\varphi,\downarrow}, \varepsilon^{\varphi,\uparrow}$ from (\ref{eq:veps1}-\ref{eq:veps2}) solve linear SDEs with exogenous driving terms. Hence, \cite[Theorem V.52]{protter.05} shows that they are given explicitly by 
\begin{align}
\varepsilon^{\varphi,\downarrow}_t &=\varepsilon^\downarrow_t+\int_0^t e^{-\int_s^t \kappa^\downarrow_r dr}\left(\frac{1-\alpha_s^\downarrow}{h_s^\downarrow}d\varphi^\downarrow_s +\frac{\alpha^\uparrow_s}{h^\uparrow_s}d\varphi^\uparrow_s\right), \label{eq:vepsexpl2}\\
\varepsilon^{\varphi,\uparrow}_t &=\varepsilon^\uparrow_t+\int_0^t e^{-\int_s^t \kappa^\uparrow_r dr}\left(\frac{1-\alpha_s^\uparrow}{h_s^\uparrow}d\varphi^\uparrow_s +\frac{\alpha^\downarrow_s}{h^\downarrow_s}d\varphi^\downarrow_s\right). \label{eq:vepsexpl1}
\end{align}

Let us now derive the wealth dynamics corresponding to a self-financing strategy $\varphi$. The self-financing condition posits that all purchases and sales have to be accounted for in the corresponding safe account $\varphi^0$. Since the order book is block shaped, the average execution price for each order is the average between the initial bid or ask price, and the one that obtains after the completion of the trade. This leads to
\begin{equation}\label{eq:phi0}
\begin{split}
d\varphi^0_t &= -\left(S^\varphi_{t-}+\varepsilon^{\varphi,\uparrow}_{t-} +\frac{1}{2h^\uparrow_t}d\varphi^\uparrow_t\right)d\varphi^\uparrow_t+ \left(S^\varphi_{t-}-\varepsilon^{\varphi,\downarrow}_{t-} -\frac{1}{2h^\downarrow_t}d\varphi^\downarrow_t\right)d\varphi^\downarrow_t\\
&=-S^\varphi_{t-}d\varphi_t -\varepsilon^{\varphi,\uparrow}_{t-} d\varphi^\uparrow_t- \varepsilon^{\varphi,\downarrow}_{t-} d\varphi^\downarrow_t -\frac{1}{2h^\uparrow_t}d[\varphi^\uparrow]_t-\frac{1}{2h^\downarrow_t}d[\varphi^\downarrow]_t.
\end{split}
\end{equation}

Now, turn to the risky position evaluated in terms of the reference price $S^\varphi$. Integration by parts as in \cite[I.4.45]{js.03} and \cite[Theorem I.4.52]{js.03} yield
\begin{align*}
d(\varphi_t S^\varphi_t) &= \varphi_{t-} dS^\varphi_t +S^\varphi_{t-} d\varphi_t+d\left[\varphi^\uparrow-\varphi^\downarrow, \frac{\alpha^\uparrow}{h^\uparrow} d\varphi^\uparrow -\frac{\alpha^\downarrow}{h^\downarrow} d\varphi^\downarrow\right]_t\\
&= \varphi_{t-} dS^\varphi_t+S^\varphi_{t-} d\varphi_t+\frac{\alpha^\uparrow_t}{h^\uparrow_t} d[\varphi^\uparrow]_t+\frac{\alpha^\downarrow_t}{h^\downarrow_t} d[\varphi^\downarrow]_t.
\end{align*}

Together, the above formulas lead to the following dynamics of the (paper) wealth process $X^\varphi=\varphi^0+\varphi S^\varphi$ corresponding to a self-financing strategy $\varphi$ with safe account $\varphi^0$ given by \eqref{eq:phi0}:

\begin{mydef}\label{def:wealth}
The \emph{(paper) wealth process} of a strategy $\varphi=\varphi^\uparrow-\varphi^\downarrow$ in the Obizhaeva/Wang model has dynamics
\begin{equation}\label{eq:wealth}
dX_t^{OW}(\varphi):= \varphi_{t-} dS^\varphi_t-\varepsilon^{\varphi,\uparrow}_{t-} d\varphi^\uparrow_t- \varepsilon^{\varphi,\downarrow}_{t-} d\varphi^\downarrow_{t} -\frac{1/2-\alpha^\uparrow_t}{h^\uparrow_t}d[\varphi^\uparrow]_t-\frac{1/2-\alpha^\downarrow_t}{h^\downarrow_t}d[\varphi^\downarrow]_t,
\end{equation}
with the bid-ask spreads $\varepsilon^{\varphi,\downarrow}, \varepsilon^{\varphi,\uparrow}$ and the reference price $S^\varphi$ from (\ref{eq:level}-\ref{eq:veps2}).
\end{mydef}

\section{Reduced-Form Models}\label{sec:AC}

Next, we introduce a general version of the reduced-form models of Almgren and Chriss~\cite{almgren.chriss.99,almgren.chriss.01,almgren.03}.\footnote{Since we want to study the connection to \emph{block-shaped} order books in the sequel, we focus on \emph{linear} temporary price impact here. See \cite{almgren.03} for an extension to nonlinear impact functions.} In this framework, trading strategies $\varphi=\varphi^\uparrow-\varphi^\downarrow$ are restricted to be absolutely continuous with finite turnover rates $\dot{\varphi}_t=\dot{\varphi}^\uparrow_t-\dot{\varphi}^\downarrow_t :=d\varphi_t/dt$, and the corresponding wealth dynamics are given by
\begin{equation}\label{eq:AC}
dX_t^{AC}(\varphi) := \varphi_t dS^\varphi_t -\varepsilon^\uparrow_t d\varphi^\uparrow_t- \varepsilon^\downarrow_t d\varphi^\downarrow_t -\lambda^\uparrow_t (\dot{\varphi}^\uparrow_t)^2 dt-\lambda^\downarrow_t (\dot{\varphi}^\downarrow_t)^2 dt.
\end{equation}
Here, as for the Obizhaeva/Wang-type model in \eqref{eq:level}, the reference price is shifted linearly by past trades:
\begin{equation}\label{eq:shift}
dS_t^\varphi:=dS_t+\gamma^\uparrow_t d\varphi^\uparrow_t-\gamma^\downarrow_t d\varphi^\downarrow_t,
\end{equation}
for continuous processes $\gamma^\uparrow,\gamma^\downarrow \geq 0$. The processes $\varepsilon^\downarrow,\varepsilon^\uparrow$ in \eqref{eq:AC} denote exogenous bid-ask spreads as in Section \ref{sec:trading}, which lead to trading costs linear in the selling and buying rates $\dot{\varphi}^\downarrow, \dot{\varphi}^\uparrow$. Finally, the positive, continuous processes $\lambda^\downarrow,\lambda^\uparrow$ levy additional quadratic trading costs on the respective turnover rates.

The dynamics \eqref{eq:AC} are again motivated by linear price impact. Indeed, suppose a single purchase $d\varphi^\uparrow_t$ executed over a time interval $dt$ moves the average transaction price by $\lambda^\uparrow_t d\varphi^\uparrow_t/dt$, i.e., price impact is linear both in the trade size and the trading speed. Then, the additional execution cost due to price impact is given by $\lambda^\uparrow_t(d\varphi_t^\uparrow/dt)^2 dt$, in line with \eqref{eq:AC}. This price impact is purely temporary, in that each trade does not affect subsequent ones through this channel. (Permanent impact is modeled separately through shifts of the reference quote \eqref{eq:shift}.) Hence, this model corresponds to a block-shaped order book that completely recovers after each trade before the next one is entered. 

This heuristic argument suggests a close connection to highly resilient order books. However, it sheds no light on how continuous trading and finite resilience interact. Making this connection precise is the main contribution of the present paper.  

\section{High-Resilience Asymptotics}\label{sec:Asymp}

We now show that wealth processes in the Obizhaeva/Wang-type model from Section~\ref{sec:trading} resemble their counterparts in a corresponding Almgren/Chriss model in the limit for large resilience, i.e., as the order book recovers faster and faster from executed trades. To this end, write the resiliences $\kappa^\uparrow,\kappa^\downarrow$ as
$$\kappa^\uparrow_t =\kappa K^\uparrow_t, \quad \kappa^\downarrow_t=\kappa K^\downarrow_t,$$
for positive, continuous, adapted processes $K^\uparrow, K^\downarrow$ independent  of the asymptotic parameter $\kappa \in \mathbb{R}_+$ that will be sent to infinity.

Consider an absolutely continuous trading strategy $d\varphi_t=\dot{\varphi}_t dt$. The following limit theorem shows that the corresponding wealth process in an Obizhaeva/Wang model with high resilience can be approximated by its counterpart in a reduced-form model of Almgren/Chriss-type:

\begin{mythm}\label{thm:main}
Fix an initial endowment $x>0$ and consider an absolutely continuous strategy $d\varphi_t=\dot{\varphi}_t dt$ with continuous turnover rate $\dot{\varphi}=\dot\varphi^\uparrow-\dot\varphi^\downarrow$. Then:
\begin{align*}
X^{OW,\kappa}(\varphi)  = x &+\int_0^\cdot \varphi_{t} dS^\varphi_t -\int_0^\cdot \left(\varepsilon^\uparrow_t \dot{\varphi}^\uparrow_t + \varepsilon^\downarrow_t \dot{\varphi}^\downarrow_t\right) dt-\kappa^{-1}\int_0^\cdot \left(\frac{1-\alpha^\uparrow_t}{K^\uparrow_t h^\uparrow_t} (\dot{\varphi}^\uparrow_t)^2 +\frac{1-\alpha^\downarrow_t}{K^\downarrow_t h^\downarrow_t} (\dot{\varphi}^\downarrow_t)^2\right) dt\\
&+o(\kappa^{-1}), \qquad \mbox{uniformly on compacts in probability, as $\kappa \longrightarrow \infty$.} 
\end{align*}
That is, as $\kappa \longrightarrow \infty$, the wealth process $X^{OW,\kappa}(\varphi)$ of $\varphi$ in the Obizhaeva/Wang model with resiliences $\kappa K^\uparrow, \kappa K^\downarrow$ resembles its counterpart $X^{AC,\kappa}(\varphi)$ in an Almgren/Chriss model with temporary quadratic impact costs $\lambda^\uparrow:=\frac{1-\alpha^\uparrow}{\kappa K^\uparrow h^\uparrow}, \lambda^\downarrow:=\frac{1-\alpha^\downarrow}{\kappa K^\downarrow h^\downarrow}=O(\kappa^{-1})$, permanent linear impact $\gamma^\uparrow:=\frac{\alpha^\uparrow}{h^\uparrow}, \gamma^\downarrow:=\frac{\alpha^\downarrow}{h^\downarrow}=O(1)$, and proportional transaction costs $\varepsilon^\uparrow, \varepsilon^\downarrow=O(1).$
\end{mythm}

To wit, the permanent price impact is independent of the resilience of the order book. In contrast, the investor's other trading costs split into two parts as the resilience grows. On the one hand, trades incur the exogenous baseline spreads. In addition, the investor is charged quadratic trading costs determined by i) the height of the order book, ii) its resilience, and iii) the asymmetry of the price impact on bid and ask spreads. This identifies the reduced-form trading cost of Almgren/Chriss in terms of observable properties of the order book.

\begin{proof}[Proof of Theorem \ref{thm:main}]
Fix $T>0$. By localization, we can assume without loss of generality that the processes $K^\uparrow$, $K^\downarrow$ are both bounded away from zero by some $\underline{K}>0$, and moreover, that the processes $\dot{\varphi}^\uparrow:=\dot{\varphi}^+$, $\dot{\varphi}^\downarrow:=\dot{\varphi}^-$ as well as $\frac{1-\alpha^{\uparrow}}{h^{\uparrow}}$, $\frac{1-\alpha^{\downarrow}}{h^\downarrow}$, $\frac{\alpha^\uparrow}{h^\uparrow}$, $\frac{\alpha^{\downarrow}}{h^\downarrow}$, $\varepsilon^\uparrow$, and $\varepsilon^\downarrow$ are all uniformly bounded by some constant $C \in (0,\infty)$. 

By Definition \ref{def:wealth}, Jensen's inequality, and the boundedness of $\dot{\varphi}^\uparrow,\dot{\varphi}^\downarrow$, it suffices to show 
\begin{equation}\label{eq:DCT}
\int_0^T \Big| \kappa (\varepsilon^{\varphi,\uparrow}_t-\varepsilon^\uparrow_t) -\frac{1-\alpha^\uparrow_t}{K^\uparrow_t h^\uparrow_t}\dot{\varphi}^\uparrow_t\Big| \mathbf{1}_{\{\dot{\varphi}^\uparrow_t>0\}} dt,\ \int_0^T \Big| \kappa (\varepsilon^{\varphi,\downarrow}_t-\varepsilon^\downarrow_t) -\frac{1-\alpha^\downarrow_t}{K^\downarrow_t h^\downarrow_t}\dot{\varphi}^\downarrow_t\Big|\mathbf{1}_{\{\dot{\varphi}^\downarrow_t>0\}}dt \longrightarrow 0, \quad \mbox{a.s.,} 
\end{equation}
as $\kappa \longrightarrow \infty$. The arguments to show this are the same for both integrals, so we focus on the first one. The explicit representation \eqref{eq:vepsexpl1}, combined with the boundedness of $\dot{\varphi}^\uparrow,\dot{\varphi}^\downarrow, \frac{1-\alpha^\uparrow}{h^\uparrow}, \frac{\alpha^\downarrow}{h^\downarrow}$, the lower bound $\underline{K}$ for $K$, and an elementary integration, implies
$$\left| \kappa (\varepsilon^{\varphi,\uparrow}_t-\varepsilon^\uparrow_t)- \frac{1-\alpha^\uparrow_t}{K^\uparrow_t h^\uparrow_t}\dot{\varphi}^\uparrow_t \right| \leq \frac{2C^2}{\underline{K}}.$$
By the dominated convergence theorem, it is thus enough to show that this integrand converges to zero pointwise for each $t \in (0,T]$. Let us therefore fix $\omega \in\Omega$ and $t \in [0,T]$. Continuity of all the involved processes implies that, for every $\delta>0$, there exists $\nu>0$ such that 
\begin{equation}\label{eq:est1}
| K^\uparrow_s - K^{\uparrow}_t| \leq \delta, \quad \mbox{for all $s \in [t-\nu,t]$},
\end{equation}
and
\begin{equation}\label{eq:est2}
\left| \frac{1-\alpha^\uparrow_s}{h^\uparrow_s}\dot{\varphi}^\uparrow_s - \frac{1-\alpha^\uparrow_t}{h^\uparrow_t}\dot{\varphi}^\uparrow_t \right| \leq \delta, \quad \mbox{for all $s \in [t-\nu,t]$.}
\end{equation}
Moreover, if $\dot{\varphi}_t>0$, the continuous trading rate is positive on $[t-\nu,t]$ (after reducing $\nu$ if necessary). Therefore, the explicit formula for the ask spread in \eqref{eq:vepsexpl1} combined with (\ref{eq:est1}-\ref{eq:est2}) gives
\begin{align*}
&\left| \kappa (\varepsilon^{\varphi,\uparrow}_t-\varepsilon^\uparrow_t) - \frac{1-\alpha^\uparrow_t}{K^\uparrow_t h^\uparrow_t}\dot{\varphi}^\uparrow_t \right| \mathbf{1}_{\{\dot{\varphi}^\uparrow_t>0\}} \\
&\qquad\leq \kappa C^2 \int_0^{t-\nu} e^{-\kappa \underline{K}(t-s)} ds + \left|\int_{t-\nu}^t \left(\kappa e^{-\kappa \int_s^t K^\uparrow_r dr} \frac{1-\alpha^\uparrow_s}{h^\uparrow_s}\dot{\varphi}^\uparrow_s \right)ds -\frac{1-\alpha^\uparrow_t}{K^\uparrow_t h^\uparrow_t}\dot{\varphi}^\uparrow_t \right|\\
&\qquad \leq \frac{C^2}{\underline{K}}e^{-\kappa\underline{K}\nu}+\left| \int_{t-\nu}^t \kappa e^{-\kappa (K^\uparrow_t-\delta)(t-s)}\left(\frac{1-\alpha_t^\uparrow}{h^\uparrow_t}\dot{\varphi}^\uparrow_t+\delta\right)ds -\frac{1-\alpha^\uparrow_t}{K^\uparrow_t h^\uparrow_t}\dot{\varphi}^\uparrow_t \right|\\
&\qquad \qquad \quad +\left| \int_{t-\nu}^t \kappa e^{-\kappa (K^\uparrow_t+\delta)(t-s)}\left(\frac{1-\alpha_t^\uparrow}{h^\uparrow_t}\dot{\varphi}^\uparrow_t-\delta\right)ds -\frac{1-\alpha^\uparrow_t}{K^\uparrow_t h^\uparrow_t}\dot{\varphi}^\uparrow_t \right|\\
&\qquad= \frac{C^2}{\underline{K}}e^{-\kappa\underline{K}\nu}+\left| \frac{\frac{1-\alpha_t^\uparrow}{h^\uparrow_t}\dot{\varphi}^\uparrow_t+\delta}{K^\uparrow_t-\delta}\left(1-e^{-\kappa(K^\uparrow_t-\delta)\nu}\right)-\frac{1-\alpha^\uparrow_t}{K^\uparrow_t h^\uparrow_t}\dot{\varphi}^\uparrow_t \right|\\
& \qquad \qquad \quad +\left| \frac{\frac{1-\alpha_t^\uparrow}{h^\uparrow_t}\dot{\varphi}^\uparrow_t-\delta}{K^\uparrow_t+\delta}\left(1-e^{-\kappa(K^\uparrow_t+\delta)\nu}\right)-\frac{1-\alpha^\uparrow_t}{K^\uparrow_t h^\uparrow_t}\dot{\varphi}^\uparrow_t \right|.
\end{align*}
Now, first choose $\delta$ sufficiently small and then $\kappa$ sufficiently large. This establishes the pointwise convergence needed to apply the dominated convergence theorem to obtain \eqref{eq:DCT} and thereby the assertion. 
\end{proof}

\begin{myrek}\label{rem:main}
In the context of utility maximization, for example, one may want to apply Theorem~\ref{thm:main} to a \emph{sequence} $(\varphi^\kappa)_{\kappa>0}$ of trading strategies that become more and more active as the market becomes more liquid for higher resilience $\kappa \longrightarrow \infty$. To this end, Theorem~\ref{thm:main} can be modified as follows: suppose
\begin{equation}\label{eq:rate}
d\varphi^\kappa_t = \kappa^{1/4} \dot{\phi}^\kappa_t dt,
\end{equation}
for continuous, adapted processes $\dot{\phi}^\kappa$, $\kappa>0$ for which the set $\{\dot{\phi}^\kappa_t: \kappa \in (0,\infty), t \in [0,T]\}$ is bounded in probability for any $T>0$,
i.e., the trading rate diverges at most at rate $\kappa^{1/4}$. 

Then, arguing verbatim as in the proof of Theorem~\ref{thm:main}, one obtains
\begin{align*}
X^{OW,\kappa}(\varphi^\kappa)  = x &+\int_0^\cdot \varphi^\kappa_t dS^{\varphi^\kappa}_t -\int_0^\cdot \left(\varepsilon^\uparrow_t \dot{\varphi}^{\kappa,\uparrow}_t + \varepsilon^\downarrow_t \dot{\varphi}^{\kappa,\downarrow}_t\right) dt\\
&-\kappa^{-1}\int_0^\cdot \left(\frac{1-\alpha^\uparrow_t}{K^\uparrow_t h^\uparrow_t} (\dot{\varphi}^{\kappa,\uparrow}_t)^2 +\frac{1-\alpha^\downarrow_t}{K^\downarrow_t h^\downarrow_t} (\dot{\varphi}^{\kappa,\downarrow}_t)^2\right) dt+o(\kappa^{-1/2}), 
\end{align*}
uniformly on compacts in probability, as $\kappa \longrightarrow \infty$.
\end{myrek}

In view of Remark \ref{rem:main}, Theorem \ref{thm:main} remains valid for families of strategies if the corresponding trading rates do not blow up too quickly. The quadratic costs levied on the trading rate in the limit are again determined by the asymmetry of the price impact, divided by the height of the order book and its resilience. Compared to Theorem \ref{thm:main}, only the rate of convergence changes.

\begin{myrek}\label{rem:main2}
The uniform convergence on compacts in probability in Theorem~\ref{thm:main} can be strengthened to uniform convergence on compacts in $L^p(P)$, $p\in[1,\infty)$ if the corresponding boundedness assumptions hold \emph{uniformly} rather than only \emph{locally}. Indeed, suppose $K^\uparrow$, $K^\downarrow$ are uniformly bounded away from zero and the processes $\dot{\varphi}$, $\frac{1-\alpha^{\uparrow}}{h^{\uparrow}}$, $\frac{1-\alpha^{\downarrow}}{h^\downarrow}$, $\frac{\alpha^\uparrow}{h^\uparrow}$, $\frac{\alpha^{\downarrow}}{h^\downarrow}$ as well as $\varepsilon^\uparrow$, $\varepsilon^\downarrow$ are all uniformly bounded. Then, the arguments in the proof of Theorem~\ref{thm:main} and another application of the dominated convergence theorem show the asserted uniform $L^p(P)$-convergence. Analogously, this also follows in the setting of Remark~\ref{rem:main}.
\end{myrek}

The limiting dynamics for absolutely continuous strategies in Theorem \ref{thm:main} suggest that block trades or other ``rougher'' strategies lead to unnecessary trading costs in this regime due to too rapid portfolio turnover. Indeed, approximating such strategies by smoothed versions as in \cite{kusuoka.95,bank.baum.04,cetin.al.04} allows for cheaper execution as the resilience grows, due to reduced price impact. This can be seen in the optimal execution paths of Obizhaeva and Wang \cite[Figure 4]{obizhaeva.wang.13}, where block trades can be gradually replaced by absolutely continuous trading rates in this case. The following lemma shows that this observation holds more generally:

\begin{mylemma}\label{lem:jump}
Fix $0<T'<T<\infty$ and consider a sequence of block trades:
$$\varphi_t= \sum_{n=0}^N \theta_n 1_{\auf \tau_n, T \zu}(t),$$
for strictly increasing stopping times $\tau_0<\tau_1<\ldots<\tau_N \leq T'$, and corresponding locally bounded block trades $\theta_n$ measurable with respect to $\scr{F}_{\tau_n}$. 

 Then, if $\varphi \neq 0$, there exists a sequence of absolutely continuous strategies $(\varphi^\kappa)_{\kappa>0}$ whose terminal payoffs dominate the one corresponding to $\varphi$ for sufficiently large resilience $\kappa$:
$${P\mbox{-}\lim}_{\kappa \to \infty} \left[X^{OW,\kappa}_T(\varphi^\kappa)-X^{OW,\kappa}_T(\varphi)\right] > 0.$$
\end{mylemma}

\begin{proof}
Approximate the simple strategy $\varphi$ by linear interpolation, with trading rates of order $O(\kappa^{1/4})$ applied over time intervals of length $O(\kappa^{-1/4})$. Then the corresponding strategies $\varphi^\kappa$ have total variation of order $O(1)$ and converge pointwise to $\varphi$ as $\kappa \longrightarrow \infty$. Moreover, by Definition \ref{def:wealth} and Remark~\ref{rem:main}, we have
\begin{equation}\label{eq:chain}
\begin{split}
X^{OW,\kappa}_T(\varphi^\kappa)-X^{OW,\kappa}_T(\varphi) \leq &\int_0^T (\varphi^\kappa_t -\varphi_t) dS_t\\
&\quad +\int_0^T \frac{\alpha^\uparrow_t}{h^\uparrow_t} \varphi^{\kappa,\uparrow}_t d\varphi^{\kappa,\uparrow}_t -\int_0^T \frac{\alpha^\uparrow_t}{h^\uparrow_t} \varphi^{\uparrow}_t d\varphi^{\uparrow}_t + \int_0^T \frac{1/2-\alpha^\uparrow_t}{h^\uparrow_t} d[\varphi^\uparrow]_t\\
&\quad -\int_0^T \frac{\alpha^\downarrow_t}{h^\uparrow_t} \varphi^{\kappa,\downarrow}_t d\varphi^{\kappa,\downarrow}_t +\int_0^T \frac{\alpha^\downarrow_t}{h^\downarrow_t} \varphi^{\downarrow}_t d\varphi^{\downarrow}_t + \int_0^T \frac{1/2-\alpha^\downarrow_t}{h^\downarrow_t} d[\varphi^\downarrow]_t\\
&\quad-\int_0^T \varepsilon^{\uparrow}_t d(\varphi^{\kappa,\uparrow}_t-\varphi^\uparrow_t)-\int_0^T \varepsilon^{\downarrow}_t d(\varphi^{\kappa,\downarrow}_t-\varphi^\downarrow_t)+O(\kappa^{-1/2}),
\end{split}
\end{equation}
in probability as $\kappa \longrightarrow \infty$, because $\varepsilon^{\varphi,\uparrow}-\varepsilon^\uparrow$, $\varepsilon^{\varphi,\downarrow}-\varepsilon^\downarrow$ are nonnegative. By localization, we can assume $\varphi$, all $\varphi^\kappa$, as well as $\alpha^\uparrow/h^\uparrow$, $\alpha^\downarrow/h^\downarrow$, $\varepsilon^\uparrow$, $\varepsilon^\downarrow$ are uniformly bounded. The first term on the right-hand side of \eqref{eq:chain} then tends to zero in probability by the dominated convergence theorem for stochastic integrals. 

Next, using the continuity of $\alpha^\uparrow/h^\uparrow$, one readily verifies that, by definition of the \emph{linear} interpolators $\varphi^\kappa$:
$$
\int_0^T \frac{\alpha^\uparrow_t}{h^\uparrow_t} \varphi^{\kappa,\uparrow}_t d\varphi^{\kappa,\uparrow}_t \longrightarrow \int_0^T \frac{\alpha^\uparrow_t}{h^\uparrow_t} \varphi^\uparrow_{t-} d\varphi^\uparrow_t +\int_0^T \frac{\alpha^\uparrow_t}{2h^\uparrow_t} d[\varphi^\uparrow]_t \quad \mbox{a.s., as $\kappa \longrightarrow \infty$.}
$$
Since $\alpha^\uparrow \in [0,1/2]$, this implies that the limit of the second line on the right-hand side of \eqref{eq:chain} is nonnegative, and strictly positive if there is at least one nontrivial purchase. 

Likewise, the third line on the right-hand side of \eqref{eq:chain} converges to a nonnegative limit as well, which is strictly positive if there is at least one nontrivial sale.

Finally, using the right-continuity of the bounded integrands and taking into account the definition of the $\varphi^\kappa$, it follows that the terms in the last line on the right-hand side of \eqref{eq:chain} tend to zero. Combining all of these estimates, the assertion follows.
\end{proof}

\section{Implications for Portfolio Choice}\label{sec:PC}

In this section, we apply the above limit theorem to reduce portfolio choice in the Obizhaeva/Wang-type models of Section~\ref{sec:trading} to the corresponding problem in the Almgren/Chriss setup from Section~\ref{sec:AC}. More specifically, we argue that the optimal trading strategies determined by Moreau et al.~\cite{moreau.al.14} in an Almgren/Chriss model with small temporary trading costs are also optimal in corresponding Obizhaeva/Wang models with high resilience.

To this end, suppose there is no exogenous baseline spread ($\varepsilon^\uparrow=\varepsilon^\downarrow=0$), and also no permanent price impact ($\alpha^\uparrow=\alpha^\downarrow=0$) so that the price process equals the fundamental value ($S^\varphi=S$). Moreover, assume that the order book is symmetric ($K^\uparrow=K^\downarrow=:K$ as well as $h^\uparrow=h^\downarrow=:h$). In summary, the wealth dynamics in the approximating Almgren/Chriss-model from Theorem~\ref{thm:main} then read as follows:
$$dX^{AC,\kappa}_t(\varphi)=\varphi_t dS_t -\frac{1}{\kappa K h}\dot{\varphi}^2_tdt.$$

Consider an agent with utility function $U: \mathbb{R} \to \mathbb{R}$ maximizing expected utility from terminal wealth at some planning horizon $T>0$.\footnote{Here, the utility function $U$ is strictly increasing, strictly concave, and twice continuously differentiable. Admissibility of corresponding trading strategies can be defined as in Biagini and {\v{C}}ern{\'y} \cite{biagini.cerny.11}, for example. Throughout, we assume that the frictionless problem is well-posed in that a unique optimal strategy $\varphi^\infty$ exists.} Write $\varphi^\infty$ for the frictionless optimizer and denote by $R$ the frictionless investor's risk-tolerance process,\footnote{That is, the risk tolerance $-u'/u''$ of the indirect utility $u$ as a function of current wealth.} which we assume to be positive, continuous, and adapted (e.g., constant for exponential utilities, see \cite{kallsen.muhlekarbe.13} for more details and intuitions). 

Given that the quadratic trading cost $\lambda=1/\kappa K h$ from Theorem~\ref{thm:main} as well as the drift and diffusion coefficients of the fundamental price process $dS_t=\mu^S_t dt+\sigma^S_t dW_t$ are sufficiently regular deterministic functions of time $t$, the current price $S_t$, and some auxiliary state variable $Y_t$ following an autonomous diffusion, \cite[Theorem 4.7]{moreau.al.14} shows that the family
 \begin{equation}\label{eq:strategy}
 d\varphi^\kappa_t= \sqrt{\frac{  \kappa K_t h_t (\sigma^S_t)^2}{2R_t}}(\varphi^\infty_t-\varphi^\kappa_t) dt, \quad \varphi^\kappa_0=\varphi^\infty_0,
 \end{equation}
is optimal \emph{at the leading order $O(\kappa^{-1/2})$}. That is, it dominates all competing families $(\vartheta^\kappa)_{\kappa>0}$ up to terms of higher order $o(\kappa^{-1/2})$:
$$E\left[U\left(X^{AC,\kappa}_T(\vartheta^\kappa)\right)\right] \leq E\left[U\left(X^{AC,\kappa}_T(\varphi^\kappa)\right)\right]+o(\kappa^{-1/2}), \quad \mbox{as } \kappa \longrightarrow \infty.$$
Here, $O(\kappa^{-1/2})$ indeed is the ``leading-order'' correction due to small price impact, because $(\varphi^\kappa)_{\kappa>0}$ achieves the optimal frictionless performance up to a term of this order \cite[Theorem 4.3]{moreau.al.14}. 

We now show that the same family $(\varphi^\kappa)_{\kappa>0}$ is also optimal in highly resilient Obizhaeva/Wang models, among competing strategies that track a sufficiently regular frictionless target with a trading rate that is at most of order $O(\kappa^{1/4})$.\footnote{Families with higher trading rates lead to trading costs of order greater than $O(\kappa^{-1/2})$ in the limiting Almgren/Chriss model, and are therefore dominated by $(\varphi^\kappa)_{\kappa>0}$ even without taking the displacement from the frictionless optimizer into account \cite[Theorem 4.3]{moreau.al.14}. Since block trades are also asymptotically suboptimal (cf.\ Lemma \ref{lem:jump}), this suggests that one can restrict to strategies of this type without losing utility at the first asymptotic order. A rigorous proof is beyond our scope here.} To this end, we first establish some properties of strategies of this type:

\begin{mylemma}\label{lem:strat}
Consider the following family of strategies $(\vartheta^\kappa)_{\kappa>0}$:
 \begin{equation}\label{eq:strategy2}
 d\vartheta^\kappa_t= \kappa^{1/2} M_t (\vartheta^\infty_t-\vartheta^\kappa_t) dt, \quad \vartheta^\kappa_0=\vartheta^\infty_0,
 \end{equation}
 for some continuous, positive, adapted trading rate $M$ and a frictionless target strategy following an It\^o process, $d\vartheta^\infty_t=\mu^{\vartheta^\infty}_tdt+\sigma^{\vartheta^\infty}_t dW_t$ with continuous drift and diffusion coefficients $\mu^{\vartheta^\infty}$, $\sigma^{\vartheta^\infty}$. 
 
 Then, the family of strategies \eqref{eq:strategy2} satisfies \eqref{eq:rate} on $[0,T]$. Moreover, 
 \begin{equation}\label{eq:DC}
 X^{AC,\kappa}(\vartheta^\kappa) \longrightarrow x+\int_0^\cdot \vartheta^\infty_t dS_t, \quad \mbox{as $\kappa \longrightarrow \infty$, uniformly in probability.}
 \end{equation}
 If the processes $\mu^{\vartheta^\infty}$, $\sigma^{\vartheta^\infty}$, and $M$ are uniformly bounded and $M$ is uniformly bounded away from zero, then \eqref{eq:rate} and the convergence in \eqref{eq:DC} hold uniformly in $L^2(P)$.
\end{mylemma}

\begin{proof}
By localization, we can assume that the continuous processes $M$, $\mu^{\vartheta^\infty}$, and $\sigma^{\vartheta^\infty}$ are bounded by some constant $C>0$ and $M$ is bounded away from zero by some constant $\underline{M}>0$. By definition of $\vartheta^\kappa$, it therefore suffices to show that $\{\kappa^{1/4}|\vartheta^\infty_t -\vartheta^\kappa_t|,\kappa\in(0,\infty), t \in [0,T]\}$ is bounded in $L^2(P)$ to establish the boundedness in probability of $\kappa^{-1/4}d\vartheta^\kappa/dt$ asserted in \eqref{eq:rate}. To see this, first notice that \cite[Theorem V.52]{protter.05} implies
$$|\vartheta^\infty_t-\vartheta^\kappa_t| = \left|\int_0^t e^{-\kappa^{1/2} \int_s^t M_r dr} d\vartheta^\infty_s\right|.$$
Hence, the algebraic inequality $(a+b)^2 \leq 2a^2+2b^2$, Jensen's inequality, Doob's maximal inequality, and the It\^o isometry yield
\begin{align*}
E\left[\sup_{t \in [0,T]} \kappa^{1/2} |\vartheta^\infty_t-\vartheta^\kappa_t|^2 \right] &\leq 2 \kappa^{1/2} C^2 T \int_0^T e^{-2\kappa^{1/2}\underline{M}(T-s)}ds
 +8 \kappa^{1/2} C^2 T \int_0^T e^{-2\kappa^{1/2}\underline{M}(T-s)}ds \\
&\leq \frac{5C^2 T}{\underline{M}}.
\end{align*}
As a result, the family $(\vartheta^\kappa)_{\kappa>0}$ satisfies \eqref{eq:rate}. Together with the dominated convergence theorem, this immediately yields the second part of the assertion. 

If all involved processes are uniformly bounded, no localization is necessary and the same arguments show uniform convergence in $L^2(P)$.
\end{proof}

We can now prove that the strategy \eqref{eq:strategy} of \cite{moreau.al.14} is not only optimal in the Almgren/Chriss model with small temporary price impact, but also in the corresponding Obizhaeva/Wang model with large resilience:

\begin{mythm}\label{thm:portfolio}
Suppose the family of strategies $(\varphi^\kappa)_{\kappa>0}$ from \eqref{eq:strategy} is optimal at the leading order $O(\kappa^{-1/2})$ in the Almgren/Chriss-model.\footnote{See \cite{moreau.al.14} for sufficient conditions. In particular, this holds if risk tolerance $R$ is constant, the drift and diffusion coefficients $\mu^S(y)$, $\sigma^S(y)$, $\mu^{Y}(y)$, $\sigma^{Y}(y)$ of the fundamental price process $S$ and the state variable $Y$ are all bounded and smooth with bounded, smooth derivatives of all orders, and the volatilities $\sigma^S(y), \sigma^Y(y)$ are bounded away from zero \cite[Theorem 8.1]{moreau.al.14}.}

Moreover, assume that the utility function $U$ has bounded absolute risk aversion $-U''/U'$, and that the classes
\begin{equation}\label{eq:UI}
\left\{ \kappa^{1/2} U'\left[X^{AC,\kappa}_T(\varphi^\kappa)\right]\left[X^{OW,\kappa}_T(\varphi^\kappa)-X^{AC,\kappa}_T(\varphi^\kappa)\right]: \kappa>0\right\}
\end{equation}
as well as
\begin{equation}\label{eq:UI2}
\left\{ \kappa^{1/2} \left(U'\left[X^{OW,\kappa}_T(\varphi^\kappa)\right]+U'\left[X^{AC,\kappa}_T(\varphi^\kappa)\right]\right) \left[X^{OW,\kappa}_T(\varphi^\kappa)-X^{AC,\kappa}_T(\varphi^\kappa)\right]^2: \kappa>0\right\}
\end{equation}
are uniformly integrable.

Then, the family $(\varphi^\kappa)_{\kappa>0}$ is also optimal at the leading order $O(\kappa^{-1/2})$ in the Obizhaeva/Wang model, among all families of strategies $(\vartheta^\kappa)_{\kappa>0}$ of the form \eqref{eq:strategy2} which satisfy \eqref{eq:UI}.
\end{mythm}

\begin{proof}
Let $(\vartheta^\kappa)_{\kappa>0}$ be any competing family of strategies. Concavity of the utility function $U$ and asymptotic optimality of $\varphi^\kappa$ in the limiting Almgren/Chriss model from Theorem \ref{thm:main} yield 
\begin{align*}
E\left[U\left(X^{OW,\kappa}_T(\vartheta^\kappa)\right)\right] &\leq E\left[U(X^{AC,\kappa}_T(\vartheta^\kappa)\right]+E\left[U'(X^{AC,\kappa}_T(\vartheta^\kappa))(X^{OW,\kappa}_T(\vartheta^\kappa)-X^{AC,\kappa}_T(\vartheta^\kappa))\right]\\
&\leq E\left[U(X^{AC,\kappa}_T(\varphi^\kappa)\right]+E\left[U'(X^{AC,\kappa}_T(\vartheta^\kappa))(X^{OW,\kappa}_T(\vartheta^\kappa)-X^{AC,\kappa}_T(\vartheta^\kappa))\right]\\
&\qquad +o(\kappa^{-1/2}).
\end{align*}
By Remark \ref{rem:main} and Lemma \ref{lem:strat}, we have $\kappa^{1/2}\left[X^{OW,\kappa}_T(\vartheta^\kappa)-X^{AC,\kappa}_T(\vartheta^\kappa)\right] \longrightarrow 0$ as $\kappa \longrightarrow \infty$, in probability. Moreover, $U'[X^{AC,\kappa}_T(\vartheta^\kappa)] \longrightarrow U'[x+\int_0^T \vartheta^\infty_t dS_t]$ in probability by Lemma \ref{lem:strat}, so that the product of these two terms converges to zero in probability. Together with the assumed uniform integrability \eqref{eq:UI}, this shows convergence in $L^1(P)$, and in turn
\begin{equation}\label{eq:bound}
E[U(X^{OW,\kappa}_T(\vartheta^\kappa)] \leq E[U(X^{AC,\kappa}_T(\varphi^\kappa)]+o(\kappa^{-1/2}).
\end{equation}

Now, consider the candidate family $(\varphi^\kappa)_{\kappa>0}$. A second-order Taylor expansion yields 
\begin{align*}
E\left[U\left(X^{OW,\kappa}_T(\varphi^\kappa)\right)\right]  &=  E\left[U(X^{AC,\kappa}_T(\varphi^\kappa)\right]+E\left[U'(X^{AC,\kappa}_T(\varphi^\kappa))(X^{OW,\kappa}_T(\varphi^\kappa)-X^{AC,\kappa}_T(\varphi^\kappa))\right]\\
&\qquad +\frac{1}{2} E\left[U''(\xi)(X^{OW,\kappa}_T(\varphi^\kappa)-X^{AC,\kappa}_T(\varphi^\kappa))^2\right],
\end{align*}
for some (random) $\xi$ between $X^{OW,\kappa}_T(\varphi^\kappa)$ and $X^{AC,\kappa}_T(\varphi^\kappa)$. As above, it follows that the first-order term is of order $o(\kappa^{-1/2})$. For the second-order term, recall that risk aversion is bounded from above ($-U''/U' \leq C$ for some $C>0$) and the marginal utility $U'$ is increasing. Whence:
\begin{align*}
&\left| E\left[U''(\xi)(X^{OW,\kappa}_T(\varphi^\kappa)-X^{AC,\kappa}_T(\varphi^\kappa))^2\right] \right|\\
&\qquad \qquad  \leq  C E\left[ \left(U'[X^{OW,\kappa}_T(\varphi^\kappa)]+U'[X^{AC,\kappa}_T(\varphi^\kappa)]\right)\left(X^{OW,\kappa}_T(\varphi^\kappa)-X^{AC,\kappa}_T(\varphi^\kappa)\right)^2\right].
\end{align*}
Again by Remark \ref{rem:main} and Lemma \ref{lem:strat}, the integrand of the right-hand side is of order $o(\kappa^{-1/2})$ in probability. Due to the uniform integrability assumed in \eqref{eq:UI2}, this convergence also holds in $L^1(P)$. As a result, 
$$E[U(X^{OW,\kappa}_T(\varphi^\kappa)] = E[U(X^{AC,\kappa}_T(\varphi^\kappa)]+o(\kappa^{-1/2}).$$
Combining this with \eqref{eq:bound}, the asserted asymptotic optimality of the family $(\varphi^\kappa)_{\kappa>0}$ follows.
\end{proof}

The prerequisites of Theorem \ref{thm:portfolio} only ask for sufficient integrability. Hence, they are satisfied automatically if all primitives of the model are uniformly bounded:

\begin{mycor}\label{cor:opt}

Suppose the resilience $K$ and the height $h$ of the order book are both uniformly bounded and bounded away from zero.

Consider a family $(\vartheta^\kappa)_{\kappa>0}$ of policies as in \eqref{eq:strategy2}. If the corresponding trading rate $M$ is bounded and bounded away from zero, and the frictionless target strategy $\vartheta^\infty$ as well as its drift and diffusion coefficients $\mu^{\vartheta^\infty}$, $\sigma^{\vartheta^\infty}$ are bounded, then the corresponding class \eqref{eq:UI} is uniformly integrable. 

If the volatility $\sigma^S$ and the frictionless risk-tolerance process $R$ are bounded and bounded away from zero, and the frictionless optimizer $\varphi^\infty$ is a bounded It\^o process with bounded drift and diffusion coefficients, then \eqref{eq:UI} and \eqref{eq:UI2} are uniformly integrable for the family \eqref{eq:strategy}.
\end{mycor}

In particular, Theorem \ref{thm:portfolio} applies in the setting of \cite[Theorem 8.1]{moreau.al.14}.\footnote{The liquidation penalty of \cite{moreau.al.14} is negligible at the leading order $O(\kappa^{-1/2})$ for strategies of the form \eqref{eq:strategy2} under the stated assumptions. Indeed, it follows from the same arguments as in the proofs of Lemma \ref{lem:strat} and Corollary \ref{cor:opt} that the penalty is of order $O(\kappa^{-1})$ in this case.}

\begin{proof}[Proof of Corollary \ref{cor:opt}]
Since the risk aversion $-U''/U'$ of the utility $U$ is bounded and bounded away from zero, $U'(x) \leq C \exp(-Bx)$ for some constants $B, C>0$ (cf.\ \cite[Remark 2.3]{kuehn.muhlekarbe.14}). Combining this with the Cauchy-Schwarz inequality, it therefore suffices to show that the families $\exp(-2B X^{AC,\kappa}_T(\vartheta^\kappa))$, 
$\kappa>0$ and $(\kappa^{1/2}[X^{OW,\kappa}_T(\vartheta^\kappa)-X^{AC,\kappa}(\vartheta^\kappa)])^2$, $\kappa>0$ are bounded in $L^1(P)$ to prove the first part of the assertion. To see this, first recall that
$$X^{AC,\kappa}_T(\vartheta^\kappa)=x+\int_0^T \vartheta^\kappa_t dS_t -\int_0^T \frac{(\dot{\vartheta}^\kappa_t)^2}{\kappa K_t h_t}dt.$$
Next notice that, since the frictionless target strategy $\vartheta^\infty$ is bounded and the trading rate $M$ is bounded and bounded away from zero, it follows from Gronwall's lemma that the family $\vartheta^\kappa$, $\kappa>0$ is uniformly bounded as well. Hence, the drift and diffusion coefficients of the It\^o processes $X^{AC,\kappa}_T(\vartheta^\kappa)$ are uniformly bounded, so that the $L^1(P)$-boundedness of $\exp(-2B X^{AC,\kappa}_T(\vartheta^\kappa))$ follows from Novikov's condition. The corresponding $L^1(P)$-bound for $(\kappa^{1/2}[X^{OW,\kappa}_T(\vartheta^\kappa)-X^{AC,\kappa}(\vartheta^\kappa)])^2$ follows from the $L^2(P)$-convergence in Remark~\ref{rem:main2}.

The above arguments apply, in particular, to the family \eqref{eq:strategy} if the volatility $\sigma^S$, the risk-tolerance process $R$, as well as the resilience $K$ and the height $h$ of the order book are uniformly bounded and bounded away from zero. In this case, the uniform integrability of \eqref{eq:UI2} for the family \eqref{eq:strategy} also follows along the same lines.
\end{proof}

 \bibliographystyle{abbrv}
\bibliography{mms}

\end{document}